\documentclass[graybox]{svmult}

\usepackage{mathptmx}       
\usepackage{helvet}         
\usepackage{courier}        
\usepackage{type1cm}        
\usepackage{makeidx}         
\usepackage{graphicx}        
\usepackage{multicol}        
\usepackage[bottom]{footmisc}
\usepackage{epstopdf}

\usepackage{subfigure}
\usepackage{url}
\usepackage{amssymb}
\usepackage{amsmath}
\usepackage{amsfonts}
\usepackage{wrapfig}
\usepackage{afterpage}
\usepackage{float}
\graphicspath{{fg/}}
\newcommand{\real}{\mathbb{R}}
\newcommand{\beps}{\mbox{\boldmath$\epsilon$}}
\newcommand{\brho}{\mbox{\boldmath$\rho$}}
\newcommand{\bsig}{\mbox{\boldmath$\sigma$}}
\newcommand{\btau}{\mbox{\boldmath$\tau$}}
\newcommand{\bxi}{\mbox{\boldmath$\xi$}}
\newcommand{\bzeta}{\mbox{\boldmath$\zeta$}}

\def\PP{P}

\def\bx{\mathbf{x}}
\def\b{\mathbf{b}}
\def\f{\mathbf{f}}

\def\sig{\sigma}

\def\aa{a}
\def\cc{c}
\def\Lam{\Lambda}
\def\vsig{\varsigma}
\def\calZ{{\cal Z}}
\def\calU{{\cal U}}

\def\half{\frac{1}{2}}
\def\t{\mathbf{t}}
\def\bt{\mathbf{t}}

\def\u{\mathbf{u}}
\def\bu{\mathbf{u}}

\def\I{\mathbf{I}}

\def\bv{\mathbf{v}}

\def\K{\mathbf{K}}

\def\bc{\mathbf{c}}
\def\bD{\mathbf{D}}
\def\bK{\mathbf{K}}
\def\RR{\mathbb{R}}
\def\Diag{\mbox{Diag}}
\def\bA{\mathbf{A}}

\def\ot{ \circ}
\def\ba{\mathbf{a}}
\def\bG{\mathbf{G}}

\def\bff{\mathbf{f}}
\def\calS{{\cal S}}
\def\calE{{\cal E}}

\def\calP{{\cal P}}

\def\bN{\mathbf{N}}
\def\dO{{{\rm d}\Omega}}
\def\dG{{{\rm d}\Gamma}}
\def\barbsig{\bar{\bsig}}
\def\barvsig{\bar{\vsig}}
\def\barbrho{\bar{\brho}}
\def\barbzeta{\bar{\bzeta}}

\def\eb{\begin{equation}}
\def\ee{\end{equation}}
\newtheorem{Algorithm}{Algorithm}

\makeindex             

\begin{document}

\title*{Canonical Duality Theory for Topology Optimization}
\author{ David Yang Gao }
\institute{David Yang Gao \at Faculty  of Science \& Technology, Federation University Australia,\\ \email{d.gao@federation.edu.au}
 }

\maketitle
\vspace{-2cm}
\begin{abstract}\\
This paper presents a   canonical duality approach for solving a general topology optimization problem of  nonlinear elastic structures. By using finite element method, this most challenging problem can be formulated as a   mixed integer nonlinear programming problem (MINLP),
  i.e. for a given deformation, the  first-level optimization is a typical linear constrained  0-1 programming problem, while for a given structure, the second-level optimization is a general nonlinear continuous minimization problem in computational nonlinear elasticity.
 It is discovered that for linear elastic structures, first-level optimization is a typical Knapsack problem, which is considered to be NP-complete
  in computer science. However, by using canonical duality theory,  this well-known  problem  can be solved analytically to obtain exact integer solution. A perturbed canonical dual algorithm (CDT) is proposed and  illustrated by benchmark problems in topology optimization.
Numerical results show that the proposed CDT method produces desired  optimal structure  without any
gray elements. The checkerboard  issue  in traditional methods is much reduced.
\end{abstract}

\section{General Topology Optimization Problem and Challenges}
Topology optimization aims to distribute materials within a prescribed design domain in order to  obtain the best structural
performance under  certain geometric or physical constraints. Due to its broad  applications, the topology optimization has been subjected to 
extensively study since the seminal paper by Bendsoe and Kikuch~\cite{bb_Bendsoe88}.
 Generally speaking, a typical topology optimization problem  involves both continuous state variable  and discrete   density distribution
 that can take either the value 0 (void) or 1 (solid material) at any point in the design domain. Thus,  numerical discretization methods  (say FEM) for solving  topology optimization problems lead to a so-called  mixed integer nonlinear programming (MINLP)  problem, which appears  extensively in computational engineering, decision and management sciences, operations research, industrial and systems engineering \cite{gao-aip16}.

Let us consider  an elastically deformable body that in an undeformed configuration occupies an open  domain $\Omega\subset \real^d \; (d=2,3)$ with (Lipschitz) boundary $\Gamma = \partial \Omega $. We assume that the body is subjected to a body force ${\bf f}$ (per unit mass) in the reference domain $\Omega$  and a given surface traction $\bt(\bx)$ of dead-load type  on the boundary ${\Gamma_t} \subset \partial \Omega $, while the body is fixed on the remaining boundary ${\Gamma_u} = \partial  \Omega\cap {\Gamma_t}$.
%
%
Based on the minimal potential principle in continuum mechanics, the topology optimization of compliance minimization problem of this elastic body can be formulated in the following coupled minimization problem
\begin{equation}
(\calP): \;\; \min_{\bu \in \calU_a}  \min_{\rho \in {\cal Z}}
\left\{ \Pi(\bu, \rho) = \int_{\Omega}
W(\nabla \bu) \rho {\rm d} \Omega + \int_{\Omega} \bu \cdot {\bf f}  \rho
{\rm d} \Omega - \int_{\Gamma_t} \bu \cdot \bt {\rm d} \Gamma
\right\} , \label{pprobm}
\end{equation}
where the unknown   $\bu : \Omega \rightarrow \real^d$ is a displacement  vector field, the design variable $\rho(\bx) \in \{ 0, 1\}$ is a discrete  scalar  field, the stored  energy per unit reference volume $W(\bD)$ is a nonlinear differentiable function of the deformation gradient $\bD=\nabla \bu$.  The notation $\calU_a$ identifies a \emph{kinematically admissible space} of deformations,  in which, certain  geometrical/boundary conditions are given, and
\[
{\cal Z} = \left\{ \rho(\bx): \Omega \rightarrow  \{ 0, 1  \} | \;\; \int_\Omega \rho(\bx) \dO \le V_c \right \}
\]
is a design feasible space, in which, $V_c > 0$ is the desired  volume.

Mathematically speaking, the topology optimization $(\calP)$ is a coupled nonlinear-discrete minimization problem in infinite-dimensional space. For large deformation problems, the stored energy $W(\bD)$ is usually nonconvex. The criticality condition of this minimization problem leads to a nonlinear system of  highly coupled partial differential equations. It is fundamentally difficult to analytically solve this type of
problems. Numerical methods must be adopted.

Finite element method is the most popular numerical approach for topology optimization, by which, the domain $\Omega$ is divided into $n$ disjointed elements $\{\Omega_e\}$ and in each element, the unknown fields can be numerically discretized as
\begin{equation}\label{eq-urho}
\bu(\bx) = \bN_e(\bx) \bu_e , \;\; \rho(\bx) = \rho_e \in \{0,  1 \} \;\; \forall \bx \in \Omega_e,
\end{equation}
where $\bN_e$ is an interpolation matrix, $\bu_e$ is a nodal displacement vector, the binary design variable $\rho_e \in \{ 0,1\}$ is used for determining whether the element $\Omega_e $ is a void ($\rho_e = 0$) or a solid ($\rho_e = 1$). Thus, by substituting (\ref{eq-urho}) into $\Pi(\bu, \rho)$ and let $\calU^m_a \subset \real^m$   be an admissible  nodal  displacement space,
\eb
 \calZ_a =
 \left \{ \brho = \{ \rho_e \} \in   \{ 0, 1\}^n |
  \;\; V(\brho) = \sum_{e=1}^n \rho_e \Omega_e \le V_c \right\},
\ee
the variational problem $(\calP)$ can be numerically reformulated the following global optimization problem
\begin{equation}
 (\calP_h): \;\;\; \min_{\bu \in \calU^m_a} \min_{\brho \in \calZ_a }
 \left\{ \Pi_h(\bu, \brho) = C(\brho,\bu)  - \bu^T \bff(\brho) \right\},
\end{equation}
where
\eb
C(\brho, \bu)
= \brho^T \bc(\bu), \;\; \bc(\bu) =  \left\{
 \int_{\Omega_e}  W(\nabla \bN(\bx) \bu_e) \dO \right\} \in \real^n, \label{eq-cu}
 \ee
\eb
\bff(\brho) =   \left\{  \int_{\Omega_e}   \rho_e \bN_e(\bx)^T\b_e(\bx) \dO   \right\}
+
\left\{ \int_{\Gamma^e_t} \bN(\bx)^T \t(\bx) \dG \right\}  \; \in \real^m.
\ee
Clearly, this discretized topology optimization  involves both the continuous variable $\bu \in \calU^m_a$ and the  integer variable $\brho \in \calZ_a$, it  is the so-called {\em mixed integer nonlinear programming problem} (MINLP) in mathematical programming. Since $\rho_e^p= \rho_e \;\; \forall \rho_e \in \{ 0, 1\}, \;\; \forall p \in \real,$ we have
\eb
  C_p(\brho, \bu) := \sum_{e=1}^n \rho_e^p c_e(\bu)= (\;
  \underbrace{ \brho \ot \dots \ot \brho }_{ p \mbox{ times }}  \;)^T \bc(\bu) =  C(\brho, \bu)
  \;\;\;  \forall p \in \real,
\ee
  where $\brho \ot \bc = \{ \rho_e c_e \} $ represents the Hadamard product.
  Particularly, for $p=2$, we write
   \eb
   C_2(\brho, \bu) := \half \brho^T \bA(\bu) \brho  , \;\;\; \bA(\bu) = 2 \Diag\{ \bc(\bu) \}.
  \ee
Clearly,  $C_2(\brho, \bu)$ is a convex function of $\brho$ since
$\bA(\bu)  \succeq 0  \;\; \forall \bu \in \calU^m_a$.
By the facts that $\brho \in \calZ_a$ is the main design variable and the displacement $\bu$ depends on each given domain $\Omega$, the problem $(\calP_h)$ is actually a so-called bi-level programming problem:
\begin{eqnarray}
 (\calP_{bl}):\;\; &  \;\;\;\;\;\; &  \min_{\brho\in \calZ_a} \min_{\bu \in \calU^m_a}
\{  C_p(\brho, \bu)   - \bu^T \bff(\brho)  \} \\ 
 &s.t. &  \bu =\arg \min_{\bv \in \calU^m_a } \Pi_h(\bv, \brho). \label{eq-llopt}
\end{eqnarray}
In this formulation, $C_p(\brho, \bu)   - \bu^T \bff(\brho) $ represents the upper-level cost function and the total potential energy $\Pi_h(\bu, \brho)$ represents the lower-level cost function. For large deformation problems, the total potential energy $\Pi_h$  is usually a nonconvex function of $\bu$. Therefore, this bi-level optimization could be the most challenging problem in global optimization.

For linear elastic structures, the total potential energy $\Pi_h$ is a quadratic function of $\bu$
\eb
\Pi_h(\bu, \brho) = \half \bu^T \bK(\brho) \bu - \bu^T \bff(\brho)
\ee
where $ \bK(\brho) = \left\{ \rho_e \bK_e \right\} \in \real^{m\times m} $ is the overall stiffness matrix, which is obtained by assembling the sub-matrix $\rho_e \bK_e$ for each element $\Omega_e$.
In this case, the lower-level optimization (\ref{eq-llopt}) is a convex minimization and for each given upper-level design variable $\brho$, the lower-level solution is simply governed  by the linear equilibrium equation
$\bK(\brho) \bu = \bff(\brho).$
 Therefore,
  the topology optimization for linear elasticity is mathematically an linear constrained
integer programming problem:
\eb\label{eq-le}
(\calP_{le}): \;\;  \min_{\brho \in \calZ_a} \min_{\bu \in \calU^m_a}
\left\{  -  \half \bu^T \bK(\brho) \bu   | \;\;\; \bK(\brho) \bu  = \bff(\brho)\right\} .
\ee

Due to the integer constraint, to solve this mixed integer quadratic minimization problem is fundamentally difficult.
 In order to overcome the combinatorics complexity  in this problem, various approximations were proposed during  the last decades, including homogenization~\cite{bb_Bendsoe88}, density-based approximations \cite{bb_Bendsoe89}, level set method \cite{bb_VanDijk2013}, and  topological derivative~\cite{bb_Sokolowski99}, etc . These approaches generally relax the MINLP problem into a continuous parameter optimization problem by using size, density or shape, and then solve it based on the traditional Newton-type (gradient-based)  or evolutionary optimization algorithms.
A comprehensive survey on these  approaches was  given in \cite{bb_Sigmund2013}.

The so-called Simplified Isotropic Material with Penalization (SIMP) is one of the most popular approaches in topology optimization:
\begin{eqnarray}\label{eq-simp}
(SIMP): \;\; & &
\min_{\brho\in\RR^N}  C_p ( \brho, \u(\brho) ) \\
& s.t. & \;\; \K( \brho^p)\u =\f( \brho ),  \;\;
V(\brho) \leq V_c, \\
& & \; \;  0< \rho_e\leq 1, \ e=1,\ldots, n
\end{eqnarray}
where $p$ is the so-called penalization parameter in topology optimization.
The SIMP formulation has been studied extensively in topology optimization and numerous research papers have been produced during the past decades. By the fact that  $ \brho^p  =  \brho \;\; \forall p \in \real, \;\; \forall \brho \in \{ 0, 1\}^n$, we can see  that the integer constraint $\brho \in \{ 0, 1\}^n$ in $(\calP_{le})$ is  simply replaced by the box constraint $\brho \in (0, 1]^n$.
Although it was discovered by engineers that  the  ``magic number" $ p = 3 $ can  ensure good convergence to almost $0$-$1$ solutions, the SIMP formulation    is not mathematically equivalent to the topology optimization problem $(\calP_{le})$. Actually, in many real-world applications, most SIMP solutions $\{ \rho_e\}$ are only approximate to $0$ or $1$ but never be  exactly $0$ or $1$.
Correspondingly, these elements are in gray scale which have to be filtered or interpreted physically. Additionally, this method suffers some key limitations such as the unsure  global optimization, many  gray scale elements  and checkerboard patterns, etc.

\section{Canonical Dual Problem and Analytical Solution}
  Canonical dual finite element methods for solving elasto-plastic structures and large deformation problems have been studied since 1988  \cite{gao-cs88,gao-jem96}. Applications to nonconvex mechanics are  given recently
for post-buckling problems \cite{ali-gao, santos-gao}. This paper  will address the canonical duality theory for solving the challenging integer programming problem in $(\calP_u)$.

Let  $\ba = \{ \aa_e = \mbox{Vol} (\Omega_e) \} \in \real^n$, where $\mbox{Vol}(\Omega_e)$ represents the volume of  each element $\Omega_e$. Then we have $\calZ_a = \{ \brho \in \{ 0, 1\}^n | \;\; \brho^T \ba \le V_c\} $.
By the fact  that
$ \min_{\brho } \min_{\bu } = \min_{\bu } \min_{\brho } $,
 the alternative iteration can be adopted for solving the topology optimization problem.
Since $ C_1(\brho,  \bu)  =  \half \bu^T  \bK(\brho) \bu  =   \brho^T   \bc(\bu)  $,
  for a given  solution of (\ref{eq-llopt}), the energy  vector
 $\bc_u = \bc(\bu) \in \real^n_+   $   is non-negative.
Thus,  the iterative method for  linear elastic topology optimization  $(\calP_{le})$ can be proposed for solving the following
linear 0-1 programming problem   ($(\calP) $ for short) :
\eb\label{eq-low}
 (\calP): \;\;  \min  \left\{ P_u(\brho) =-  \bc_u^T   \brho  \;\; 
  | \;\;\brho \in \{ 0, 1\}^n,  \;\; \brho^T \ba \le V_c \right\}.
\ee
This is the well-known  Knapsack problem. Due to the 0-1 constraint,
even this most simple linear integer programming  is listed as one of Karp's 21 NP-complete problems~\cite{karp}.
 However, this challenging problem can be solved analytically by using the canonical duality theory.

 The canonical duality theory for  general integer programming   was first proposed by Gao
in 2007 \cite{gao-jimo07}. The key idea of this theory is the introducing of a canonical measure
\eb
\bxi=  \Lam(\brho) = \{
\brho \ot  \brho - \brho, \;\; \brho^T \ba -V_c \} : \;\; \real^n \rightarrow
  \calE =  \real^{n+1}.
  \ee
  Let
  \eb
  \calE_a := \{ \bxi= \{ \beps, \nu\} \in \real^{n+1}|\;\; \beps \le 0, \;\; \nu \le 0 \}
\ee
be a convex cone in $\real^{n+1}$. Its indicator $\Psi(\bxi)$  is defined by
\[
\Psi(\bxi) = \left\{ \begin{array}{ll}
0 & \mbox{ if } \bxi \in \calE_a \\
+\infty & \mbox{ otherwise}
\end{array}
\right.
\]
which  is a convex and   lower semi-continuous (l.s.c) function in $\real^{n+1}$. By this function, the primal problem can be relaxed in the following unconstrained minimization form:
\eb
\min  \left\{ \Phi(\brho) = P_u(\brho )  
+ \Psi(\Lam(\brho)) \;\;
  | \;\;  \brho   \in \real^n  \right\}.
\ee
Due to the convexity of   $\Psi(\bxi)$, its conjugate function can be defined uniquely  by the Fenchel transformation:
\eb
\Psi^*(\bzeta) = \sup_{\bxi \in \real^{n+1}} \{ \bxi^T \bzeta - \Psi(\bxi) \}
=\left\{ \begin{array}{ll}
0 & \mbox{ if } \bzeta \in \calE_a^*\\
+\infty & \mbox{ otherwise}
\end{array} \right.
\ee
where
$\calE_a^* = \{ \bzeta = \{ \bsig, \vsig\} \in \real^{n+1} | \;\; \bsig \ge 0, \;\; \vsig \ge 0 \} $
is the  dual space of $\calE_a$. Thus, by using the Fenchel-Young equality $\Psi(\bxi) + \Psi^*(\bzeta) = \bxi^T \bzeta$, the function $\Phi(\brho)$ can be written in the  Gao-Strang total complementary function \cite{gs-89}
 \eb
 \Xi(\brho, \bzeta) = P_u(\brho ) 
 +  \Lam(\brho)^T \bzeta - \Psi^*(\bzeta).
 \ee

Based on this function, the canonical dual of $\Phi(\brho)$ can be defined by
 \eb
 \Phi^d(\bzeta) = \mbox{sta } \{ \Xi(\brho, \bzeta) | \;\; \brho \in \real^m  \} =
 P^\Lam_u(\bzeta ) - \Psi^*(\bzeta)
 \ee
where $\mbox{sta } \{ f(x) |\; x \in X \} $ stands for finding a stationary value of
$ f(x) \;\; \forall x\in X$, and
 \eb
P^\Lam_u(\bzeta ) = \mbox{sta } \{
   \Lam(\brho)^T \bzeta +   P_u(\brho )  \} = - \frac{1}{4} \btau_u^T(\bzeta ) \bG^{-1}(\bzeta)  \btau_u(\bzeta ) - \vsig V_c
   \ee
is the $\Lam$-conjugate of $ P_u(\brho )$, in which,
\[
\bG (\bzeta) = \ \Diag\{ \bsig\}  , \;\;\;  \btau_u( \bzeta ) =    \bsig - \vsig \ba +  \bc_u .
\]
Clearly, $P_u^\Lam(\bzeta)$ is well-defined if $ \det \bG  \neq 0  $, i.e.  $\bsig \neq 0 \in \real^n $.
 Let $\calS_a = \{ \bzeta \in \calE^*_a  | \;\; \det \bG  \neq 0 \}$.
We have the  following standard result in the canonical duality theory:
\begin{theorem}[Complementary-Dual Principle] For a  given
 $\bu \in \calU^m_a$, if $(\barbrho, \barbzeta) $ is a KKT point of $\Xi$, then
$\barbrho$ is a KKT point of $\Phi$, $\barbzeta$ is a KKT point of $\Phi^d$, and
\eb
\Phi(\barbrho) = \Xi(\barbrho, \barbzeta)=\Phi^d(\barbzeta).\label{eq-cdp}
\ee
\end{theorem}
\begin{proof}
By the convexity of  $\Psi(\bxi)$, we have the following canonical duality relations:
 \eb \label{eq-cdr}
 \bzeta \in \partial \Psi(\bxi) \;\; \Leftrightarrow \;\; \bxi \in \partial \Psi^*(\bzeta)
 \;\; \Leftrightarrow \;\; \Psi(\bxi) + \Psi^*(\bzeta) = \bxi^T \bzeta,
 \ee
 where
 \[
 \partial \Psi(\bxi) = \left\{ \begin{array}{ll}
 \bzeta & \mbox{ if } \bxi \in\calE_a \\
 \emptyset & \mbox{ otherwise}
 \end{array} \right.
 \]
 is the sub-differential of $\Psi$. Thus, in terms of $\bxi = \Lam(\brho)$ and $\bzeta = \{ \bsig, \vsig\}$, the canonical duality relations (\ref{eq-cdr}) can be  equivalently written as
 \eb
 \brho \ot \brho - \brho \le 0   \;\; \Leftrightarrow \;\; \bsig \ge 0 \;\;
 \Leftrightarrow \;\; \bsig^T ( \brho \ot \brho - \brho )  = 0\;\; \label{eq-kkts}
 \ee
 \eb
 \brho^T \ba - V_c \le 0 \;\; \Leftrightarrow \;\; \vsig \ge 0 \;\;\Leftrightarrow \;\;
 \vsig(\brho^T \ba - V_c) = 0. \label{eq-kktv}
 \ee
These are exactly the KKT conditions for the inequality constraints $\brho \ot \brho - \brho \le 0  $ and $ \brho^T \ba - V_c \le 0$. Thus,   $(\barbrho, \barbzeta) $ is a KKT point of $\Xi$ if and only if
$\barbrho$ is a KKT point of $\Phi$, $\barbzeta$ is a KKT point of $\Phi^d$. The equality (\ref{eq-cdp}) holds due to the canonical duality relations in (\ref{eq-cdr}).
\hfill $\Box$
\end{proof}

Indeed, on the effective domain $\calE^*_a$ of $\Psi^*(\bzeta)$, the total complementary function
$\Xi$ can be written as
 \eb
 \Xi(\brho, \bsig, \vsig) = P_u(\brho ) + \bsig^T ( \brho \ot \brho - \brho) + \vsig (\brho^T \ba - V_c),
 \ee
which  can be considered as the Lagrangian of $(\calP)$ for the canonical constraint
$\Lam(\brho) \le 0 \in \real^{n+1}$.
The Lagrange multiplier $\bzeta = \{ \bsig, \vsig \} \in \calE^*_a$ must satisfy the KKT conditions in (\ref{eq-kkts}) and (\ref{eq-kktv}). By the complementarity condition
 $\bsig^T ( \brho \ot \brho - \brho ) = 0$ we know that $\brho \ot \brho = \brho  $ if $\bsig >  0$.
Let
\eb
\calS_a^+ = \{ \bzeta = \{ \bsig, \vsig\}\in  \calE_a^* | \;\;    \bsig >  0 \}.
\ee
Then for any given  $ \bzeta= \{ \bsig, \vsig\}  \in \calS^+_a$,  the function $\Xi(\cdot, \bzeta): \real^m \rightarrow \real$ is strictly convex,
the canonical dual function of $\PP_u$ can be well-defined by
\eb\label{eq-Pd}
\PP^d_u(\bzeta) = \min_{\brho \in\real^m} \Xi(\brho, \bzeta)
= - \frac{1}{4} \btau_u^T(\bzeta ) \bG^{-1}(\bzeta)  \btau_u(\bzeta )  - \vsig V_c .
\ee
Thus, the canonical dual problem of $(\calP)$ can be proposed as the following:
\eb
(\calP^d): \;\;\; \max \{ \PP^d(\bsig,\vsig) | \;\; (\bsig, \vsig) \in \calS^+_a \}.
\ee

\begin{theorem}[Analytical Solution]\label{thm-rho}
For any given $\bu \in \calU^m_a$, if $  \barbzeta $ is a  solution to $(\calP^d)$, then
\eb\label{eq-solu}
\barbrho =
\half \bG^{-1} (\barbzeta)  \btau_u(\barbzeta )
\ee
is a global optimal solution to $(\calP)$ and
\eb
\PP_u(\barbrho) = \min_{\brho \in \real^n} \PP_u(\brho) = \max_{\bzeta \in \calS^+_a } \PP^d_u(\bzeta)
=  \PP^d_u(\barbzeta).
\ee
\end{theorem}
\begin{proof}
 It is easy to prove that   for any given  $\bu \in \calU^m_a$, the canonical dual function
 $\PP^d_u(\bzeta)$ is concave on the open convex set $\calS^+_a$.
 If $\barbzeta$ is a KKT point of $\PP^d_u(\bzeta)$, then it must be a unique global maximizer of
 $\PP^d_u(\bzeta)$ on $\calS^+_a$.
 By Theorem 1   we know that if
 $\barbzeta = \{ \barbsig , \barvsig \} \in \calS^+_a$ is a KKT point of $\Phi^d(\bzeta)$, then
 $\barbrho = \brho(\barbzeta)$ defined by (\ref{eq-solu}) must be a KKT point of $\Phi(\brho)$.
 Since $\Xi(\brho, \bzeta)$ is a saddle function on $\real^n \times \calS^+_a$,
we have
\begin{eqnarray*}
\min_{\brho \in \real^n} \Phi(\brho) &=&\min_{\brho \in \real^n} \max_{\bzeta \in \calS^+_a}
\Xi(\brho, \bzeta) =   \max_{\bzeta \in \calS^+_a}\min_{\brho \in \real^n}
\Xi(\brho, \bzeta)\\
& =& \max_{\bzeta \in \calS^+_a} \Phi^d(\bzeta) =  \max_{\bzeta \in \calS^+_a} \PP^d_u(\bzeta) ,
\end{eqnarray*}
Since $\barbsig > 0$,  the complementarity condition in (\ref{eq-kkts}) leads to
 \[
 \barbrho \ot \barbrho - \barbrho = 0 \;\;  \mbox{ i.e. } \barbrho \in \{ 0, 1 \}^n.
 \]
 Thus, we have
 \[
 \PP_u(\barbrho) = \min_{\brho \in \calZ_a } \PP_u(\brho) = \max_{\bzeta \in \calS^+_a } \PP^d_u(\bzeta)
 = \PP^d_u(\barbzeta)
 \]
 as required. \hfill $\Box$
 \end{proof}

 \begin{remark} Theorem \ref{thm-rho} shows that although the canonical dual problem is a concave maximization in continuous space,  it produces the analytical solution  (\ref{eq-solu})
to  the well-known integer Knapsack problem $(\calP_u)$! This analytical solution was first obtained  by Gao in 2007
for  general quadratic integer programming problems (see Theorem 3, \cite{gao-jimo07}).
 The indicator function of a convex set and its sub-differential were
 first introduced by J.J. Moreau in 1968  in his study on unilateral constrained problems
 in contact mechanics \cite{moreau68}. His pioneering work laid a foundation for modern
 analysis and the canonical duality theory.
  In solid mechanics, the indicator of a plastic yield condition is also called a
 {\em super-potential}. Its sub-differential leads to a general constitutive law
 and a unified pan-penalty finite element method in plastic limit analysis \cite{gao-cs88}.
 In mathematical programming, the canonical duality
 leads to a unified framework for nonlinear constrained
 optimization problems in multi-scale systems \cite{gao-dual00,gao-cace09,gao-aip16,gao-ruan-jogo10}.
\end{remark}

\section{Perturbed Canonical Duality Method and Algorithm}
Numerically speaking, although the global optimal solution of the integer programming problem $(\calP)$
can be  obtained by solving the canonical dual problem $(\calP^d)$,
the rate of convergence is very slow since  $\PP^d_u(\bsig,\vsig) $ is
nearly a linear function of $\bsig \in \calS^+_a$ when $\bsig $ is far from its origin. 
In order to overcome this  problem, a so-called $\beta$-perturbed canonical dual method
has been proposed  by Gao and Ruan in integer programming \cite{gao-ruan-jogo10},
i.e. by introducing a perturbation parameter $\beta > 0$,
 the problem $(\calP^d)$ is replaced by
 \eb
 (\calP^d_\beta): \;\; \max \left\{ \PP^d_\beta(\bsig, \vsig)  = \PP^d_u (\bsig, \vsig) -
  \frac{1}{4}  \beta^{-1} \bsig^T  \bsig | \;\; \{ \bsig, \vsig\} \in \calS^+_a \; \right\}
  \ee
  which is strictly concave on $\calS^+_a$.

\begin{theorem}
For a  given $\bu \neq 0 \in \real^m$ and $V_c > 0$,
there exists a $\beta_c > 0$ such that for any given $\beta \ge \beta_c$, the problem
$(\calP^d_\beta)$ has a unique solution $\bzeta_\beta \in \calS^+_a$.
If $\brho_\beta = \half \bG^{-1} (\bzeta_\beta) \btau_u(\bzeta_\beta ) \in\{ 0, 1\}^n$,
then  $\brho_\beta$ is a global optimal solution to $(\calP)$.
\end{theorem}
\begin{proof} It is easy to show that for any given $\beta > 0$,
$\PP^d_\beta(\bzeta)$ is strictly concave on the open convex set $\calS^+_a$, i.e. $(\calP^d_\beta)$
has a unique solution. Particularly,
  the criticality condition $\nabla \PP^d_\beta(\bzeta) = 0 $
leads to the the following
canonical dual algebraic equations:
\eb
4  \beta^{-1} \sig_e^3 + \sig_e^2 = (\vsig \aa_e - \cc_e)^2, \;\; e = 1, \dots, n, \label{eq-cdas}
\ee
\eb
\sum_{e=1}^n \half \frac{\aa_e}{\sig_e} ( \sig_e - \aa_e   \vsig + \cc_e) - V_c = 0 .\label{eq-cdv}
\ee
It was proved in \cite{gao-dual00} that for any given
$\beta > 0$ and $\theta_e = \vsig \aa_e - \cc_e  \neq 0, \;\;  e = 1, \dots, n$,
 the canonical dual algebraic equation (\ref{eq-cdas}) has a unique
 positive real solution
 \eb
\sigma_e  =  \frac{1}{6} \beta   [- 1 +  \phi_e(\vsig  ) + \phi_e^c(\vsig  )] > 0 , \;\; e = 1, \dots, n
\label{eq-solus}
\ee
where
\[
\phi_e(\vsig )  = \eta^{-1/3} \left[2 \theta_e^2 - \eta + 2 i  \sqrt{ \theta_e^2(\eta - \theta_e^2)}\right]^{1/3} ,
 \;\; \eta = \frac{\beta^2}{27 },
\]
and $\phi_e^c $ is   the complex conjugate of $\phi_e $, i.e.
 $\phi_e  \phi_e^c  = 1  $.
 Thus, the canonical dual algebraic equation (\ref{eq-cdv}) has a unique solution
 \eb\label{eq-soluvs}
 \vsig = \frac{\sum_{e= 1}^n \aa_e ( 1+ \cc_e/\sig_e) - 2 V_c}{\sum_{e=1}^n \aa_e^2/\sig_e}  .
 \ee
This shows that the perturbed canonical dual problem $(\calP^d_\beta)$ has a unique solution in $\calS^+_a$, which can be analytically obtained by (\ref{eq-solus}) and (\ref{eq-soluvs}).
The rest proof of this theorem is similar to that given in \cite{gao-ruan-jogo10}.
\hfill $\Box$
\end{proof}


Theoretically speaking, for any given $V_c < V_o$, the perturbed canonical duality method can produce desired  optimal solution to the integer constrained problem $(\calP)$. However, if $V_c \ll V_o$, to  reduce the initial volume $V_o$ directly to $V_c$  by solving the bi-level topology optimization problem $(\calP_{bl})$ may lead to unreasonable solutions. In order to resolve this problem, a volume decreasing  control parameter $\mu \in ( V_c/V_o,  1)$  is  introduced to slowly reduce the volume in the iteration. Thus, based on the above strategies, the canonical duality algorithm (CDT) for solving the general topology optimization problem $(\calP_{bl})$ can be proposed  below.

\begin{Algorithm}[Canonical Dual Algorithm for Topology Optimization (CDT)] $\;$ \\
{\em
\begin{verse}
(I)  Initialization.
Let $\brho^0 = \{1\} \in \real^n$.
Find $\u^0$ by solving the sub-level optimization problem
\eb
  \u^0  = \arg \min\{ \Pi(\bu, \brho^0)  | \;\; \bu \in \calU_a \}.
\ee
Compute $\bc^0 = \bc(\bu^0)$  according to~\eqref{eq-cu}. Define an initial value  $\vsig_0 > 0$ and an initial volume  $ V_\gamma \in [V_c,  V_o)$. Let $\gamma = 0, \; k=1$.\\

(II)  Find $\bsig_k = \{\sigma_e^{k } \} \in \real^n$ by
\[
\sigma_e^{k } =  \frac{1}{6} \beta   [- 1 +  \phi(\vsig^{k-1 } ) + \phi^c(\vsig^{k-1 } )] , \;\; e = 1, \dots, n.
\]

(III) Find $\vsig^{k }$ by
\[
\vsig^{k } = \frac{ \sum_{e=1}^n  \aa_e (1 +  \cc^\gamma_e /\sigma_e^{k} ) - 2 V_\gamma }{\sum_{e=1}^n \aa_e^2/  \sigma_e^{k}.    }
\]

(IV) Find $\brho^k$ by
\[
\rho^{k }_e   = \frac{1}{2} [ 1 - ( \vsig^{k } \aa_e - \cc^\gamma_e)/\sigma_e^k],
\;\; e= 1, \dots, n.
\]

(V) If
\[
|C(\brho^k, \u^\gamma)  - C( \brho^{k-1}, \u^\gamma) |  \le \omega_1,
 \]
and $\sum_{e=1}^n \rho_e^k \aa_e  \le  V_\gamma$, let $\brho^{\gamma} = \brho^k$, go to (VI); otherwise,  let $k=k+1$, go to (II).\\

(VI)
Find $\u^{\gamma}$ by solving
\eb
 \u^{\gamma} = \arg \min \{ \Pi(\bu, \brho^\gamma) | \;\; \bu \in \calU_a \} \label{eq-ugamma}
\ee

(VII)  Convergence test:
If
  \[
  |C(\brho^\gamma, \u^\gamma)- C(\brho^{\gamma-1}, \u^{\gamma-1}) | \le \omega_2, \;\;   V_\gamma  \le V_c
  \]
then stop; \\
otherwise, let $V_{\gamma+1} = \mu V_\gamma \ge V_o$ and computing
  $\bc^{\gamma+1} = \bc(\bu^\gamma)$,
Let  $\gamma = \gamma+1$, $k = 1$, go to (II).
\end{verse}
}
\end{Algorithm}
The penalty parameter in this algorithm is usually taken  $\beta > 10$.  
For linear elastic materials, the lower-level optimization (\ref{eq-ugamma}) in the algorithm (CDT) can be simply replaced by
$
u^{\gamma} = \bK^{-1} (\brho^\gamma) \bff(\brho^\gamma).
$
\section{Numerical Examples for Linear Elastic Structures}
The proposed semi-analytic method  is  implemented in Matlab.  
For the purpose of  illustration, the  applied  load and geometry data are chosen as dimensionless.
Young's modulus and Poisson's ratio of the  material are taken as $E = 1$ and $\nu = 0.3$, respectively.
The volume fraction is $\mu_c = V_c/V_0 = 0.6$. 
The stiffness matrix of the structure in CDT algorithm  is given by
$\bK (\brho) = \sum_{e=1}^n [ E_{min}  + (E-E_{min} ) \rho_e ]\bK_e$  where 
  $E_{min} = 10^{-9}$ in order to avoid singularity in computation..
The evolutionary rate used in the CDT is $\mu=0.975$.
To compare with the SIMP approach, the well-known   88-line algorithm proposed by Andreassen et al~\cite{Andreassen2011} is used with  the 
parameters  penal $=3$,
rmin = 1.5,  ft=1.

\subsection{MBB  Beam Problem}

The   well-known benchmark Messerschmitt-B\"{o}lkow-Blohm (MBB) beam problem in topology optimization
  is selected  as the first  test example    (see  Fig.~\ref{example-mbb}).
The design domain is discretized with $180 \times 60$ square mesh elements. 
Computational results obtained by both CDT and SIMP are reported in  Tables~\ref{example-mbb-nofilter}.
\begin{figure}[H]
  \centering
  \includegraphics[width=0.6\textwidth]{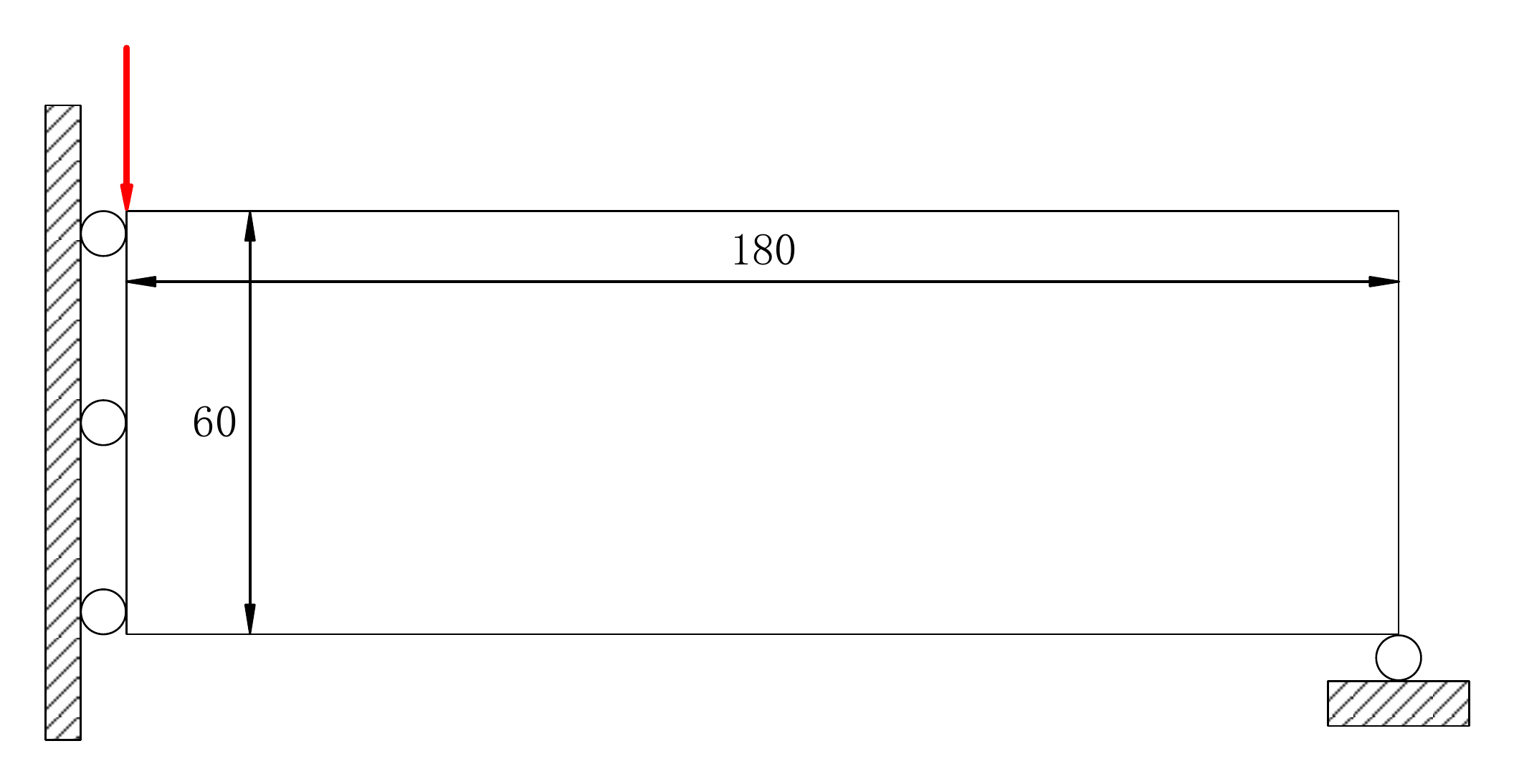}
  \caption{The design domain, boundary conditions and external load  for   a MBB beam}
  \label{example-mbb}
\end{figure}

\begin{table}
\centering
\caption{The comparison between the SIMP and CDT.}
\label{example-mbb-nofilter}
\begin{tabular}{|c|c|c|c|}
\hline
Method & Structures & Steps & Compliace \\
\hline
SIMP & \includegraphics[width=0.6\textwidth]{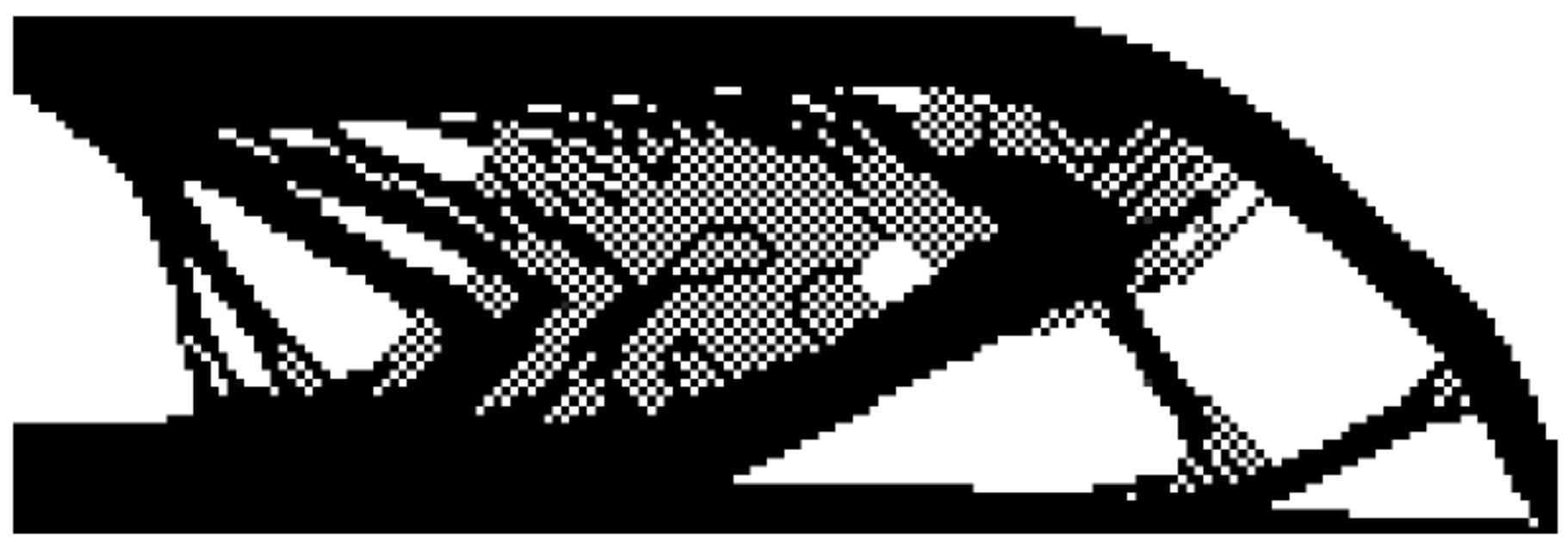} & 41 &169.2908 \\
\hline
CDT  & \includegraphics[width=0.6\textwidth]{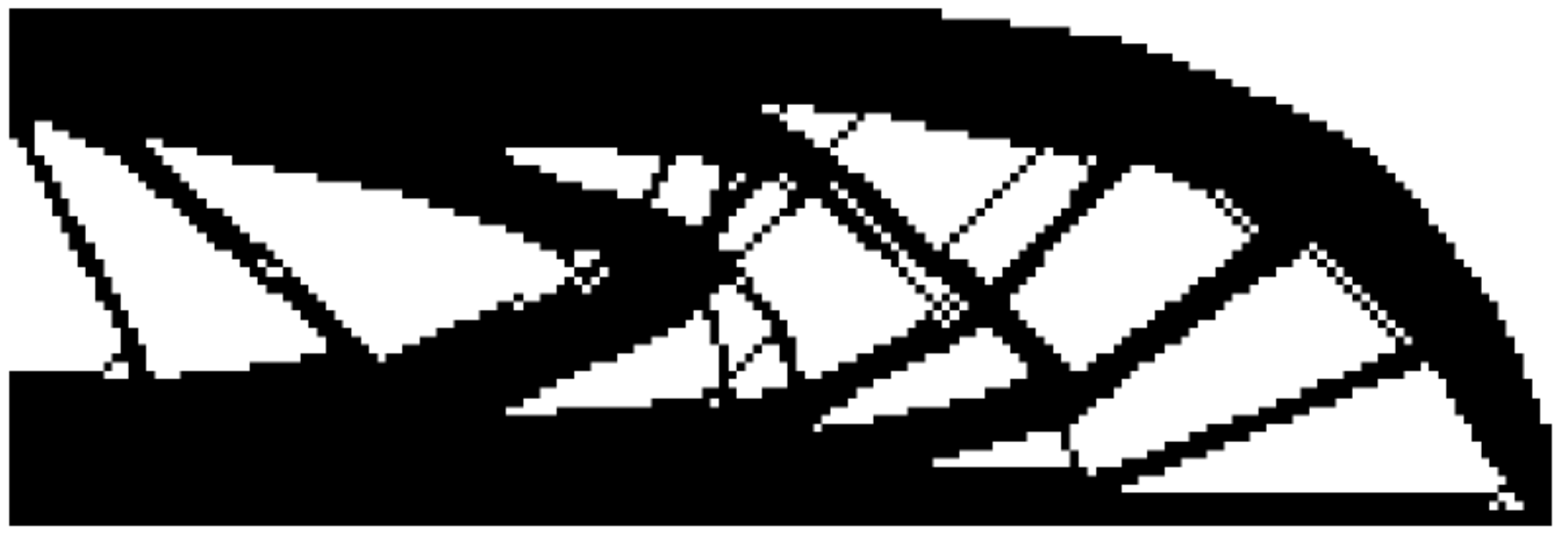}  & 28 & 164.7108 \\
\hline
\end{tabular}

\end{table}

\subsection{ Cantilever Beam}

The second  test example is the  classical Cantilever problem (see Figure~\ref{fig_CT_problem}).
 The beam  is fixed along its left side with a downward traction  applied at its right middle point.
  The example consists of $180\times 60$ quad meshes and the target volume fraction is $\mu_c = 0.6$.
  Numerical  results by both the CDT  and  SIMP   are shown in Figure~\ref{fig_CT_results}.
   
\begin{figure}
  \centering
  \includegraphics[width=0.7\textwidth]{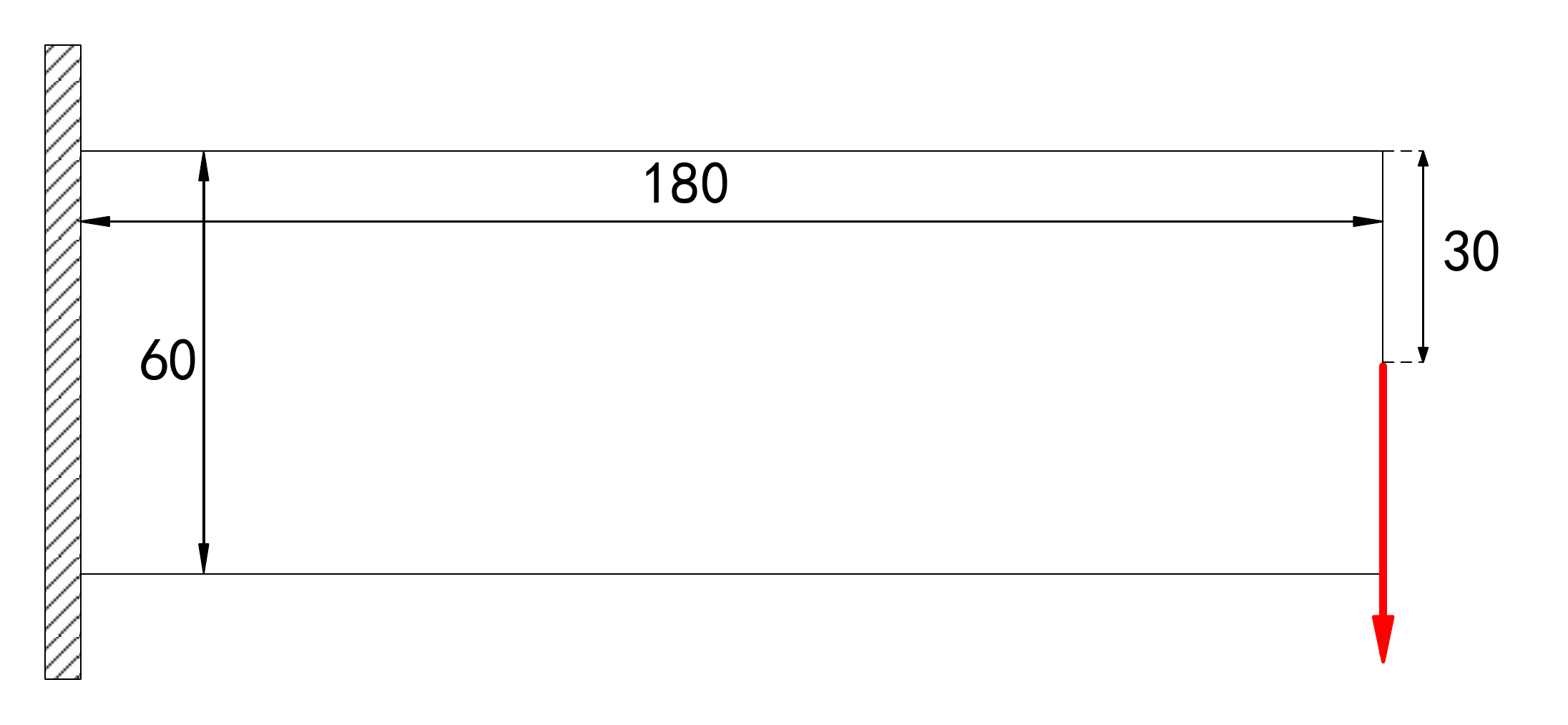}
  \caption{A test example of the benchmark Cantilever problem at volume fraction of 0.5.}
  \label{fig_CT_problem}
\end{figure}

\begin{figure}
  \centering
   \subfigure[SIMP without filter: compliance = 152.7490 with 37 iterations]{\includegraphics[width=0.7\textwidth]{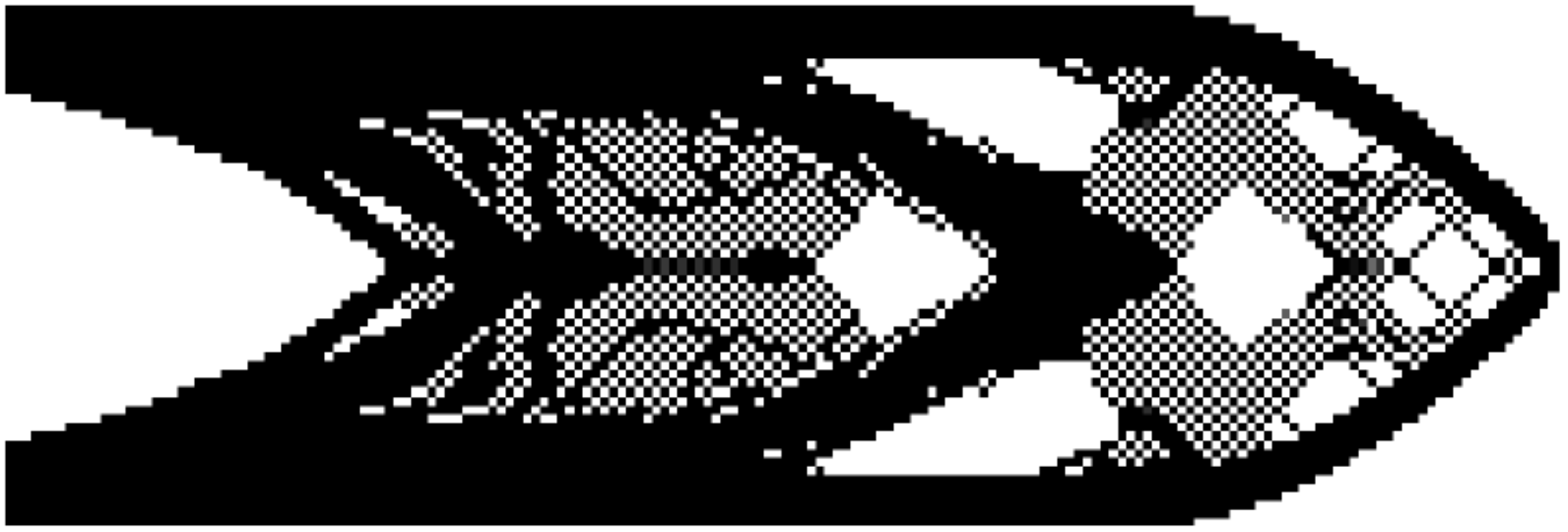}} 
   \subfigure[CDT: compliance = 153.6767 with 23 iterations]{\includegraphics[width=0.7\textwidth]{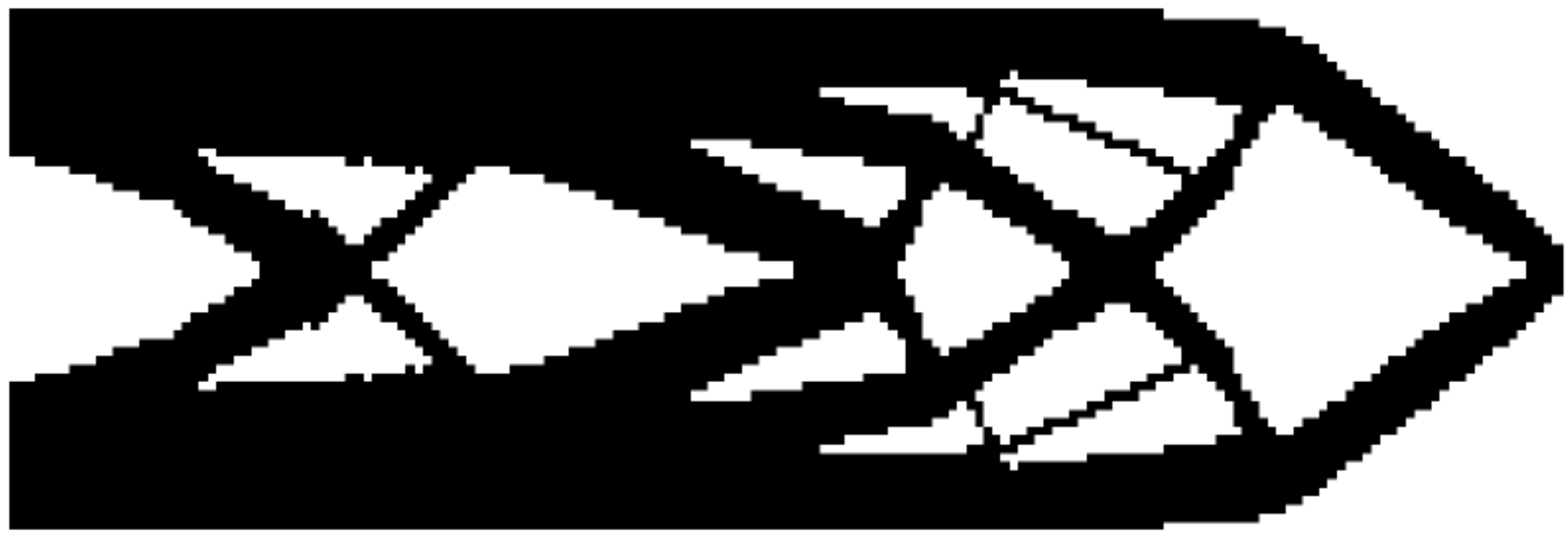}}

  \caption{ Topology optimization for the cantilever beam by the   SIMP  (a) and  CDT   (b) methods. }\label{fig_CT_results}
\end{figure}

\subsection{Summary of Computational Results} 
The computational results for the above benchmark problems show clearly
that without   filter,  the  SIMP produces a large range of checkerboard patterns and gray elements,
 while by  the CDT method, precise void-solid optimal structure can be obtained with very few checkerboard patterns.
  By the fact that the optimal density distribution $\brho$ can be obtained analytically at each iteration, the CDT method produces desired optimal structure within much less computing time.  
The convergence of the CDT method depends mainly on the parameter $\mu \in [\mu_c, 1)$. Generally speaking, the smallar  $\mu $
produces fast convergent   but   less optimal results. Detailed study on this issue will be addressed in the future research.  

\begin{acknowledgement}
MATLAB code for the CDT algorithm was helped by Professor M. Li from Zhejiang University.
The  research  is supported by US Air Force  
 Office of  Scientific Research
 under grants FA2386-16-1-4082 and FA9550-17-1-0151.
\end{acknowledgement}


\end{document}